\documentclass[letterpaper, 10pt, conference]{ieeeconf}

\IEEEoverridecommandlockouts
\usepackage[nocompress]{cite}
\usepackage{amsmath}
\usepackage{amsfonts}
\usepackage{xcolor}
\usepackage{pgfplots}
\pgfplotsset{compat=1.12}

\newtheorem{definition}{Definition}
\newtheorem{proposition}{Proposition}
\newtheorem{theorem}{Theorem}
\newtheorem{remark}{Remark}

\newtheorem{example}{Example}

\renewcommand{\epsilon}{\varepsilon}
\renewcommand{\theta}{\vartheta}

\newcommand*{\tran}{^{\mkern-1.5mu\mathsf{T}}}

\newcommand{\dd}[2]{\frac{\partial{#1}}{\partial{#2}}}

\IEEEtriggeratref{9}

\begin{document}

\title{\LARGE \bf Control Contraction Metrics on Finsler Manifolds}

\author{Thomas L. Chaffey$^1$ and Ian R. Manchester$^1$%
\thanks{$^1$The authors are with the Australian Centre for Field Robotics (ACFR),
University of Sydney, NSW 2006, Australia. Emails: \texttt{tcha4856@uni.sydney.edu.au},
\texttt{i.manchester@acfr.usyd.edu.au}.}}

\maketitle

\begin{abstract}
	Control Contraction Metrics (CCMs) provide a nonlinear controller design
	involving an offline search for a Riemannian metric and an online search for
	a shortest path between the current and desired trajectories.  In this paper,
	we generalize CCMs to Finsler geometry, allowing the use of non-Riemannian
	metrics.  We provide open loop and sampled data controllers.  The sampled
	data control construction presented here does not require real time
	computation of globally shortest paths, simplifying computation.
\end{abstract}

\section{Introduction}\label{intro}

Control synthesis for general nonlinear systems remains a challenging problem, and no
one technique is recognised as universally applicable \cite{Isidori1995, Krstic1995,
Sepulchre1997, Kokotovic2001}.  Two popular classes
of solution are explicit control constructions based on classical Lyapunov theory and
model predictive techniques involving real time optimisation.

A Lyapunov function characterizes
the stability of a system and is related to the intuitive idea of energy decaying
in stable systems.  A control Lyapunov function (CLF) is necessary and sufficient for
controllability of a system \cite{Artstein1983, Sontag1983}, and for large classes of
systems (those affine in controls), the construction of a controller given a CLF is 
simple \cite{Sontag1989}.  However, control Lyapunov functions are in general 
difficult to find \cite{Rantzer2001}.

The Control Contraction Metric (CCM) method of control synthesis, introduced by Manchester
and Slotine \cite{Manchester2017}, simplifies the search for a Lyapunov function.  
Rather than explicitly search for a Lyapunov function, a
convex search is performed for a CCM which measures distance between trajectories.
The CCM may be thought of as inducing an infinite family of local Lyapunov functions.
Online computation involves a search for a minimal path and integration of a local,
differential control law along this minimal path.  CCM controllers have more
efficient online computation than nonlinear Model Predictive Control
\cite{Leung2017}, and have been applied in several application areas, including
mechanical systems \cite{manchester2018unifying}, decentralized and distributed
control \cite{Shiromoto2016, Wang2017, Manchester2017a} and collision-free motion planning
\cite{singh2017robust}.

CCMs are based closely on contraction analysis, introduced by Lohmiller and Slotine
\cite{Lohmiller1998}.   The central idea is 
that if all nearby trajectories converge to each other, then all trajectories
converge to one nominal trajectory and the system is stable. The idea that global
properties can be inferred from the local behaviour of trajectories does away with
the need to construct global functions.  

CCMs are inherently Riemannian - the Lyapunov function locally induced by a CCM is a
Riemannian metric, similar to a traditional global quadratic Lyapunov function.  
The restriction to Riemannian metrics precludes the use of certain desirable 
non-Riemannian Lyapunov functions, for
example $p$-norms with $p \neq 2$ and consensus algorithms based upon the Hilbert
projection metric \cite{Sepulchre2010}.

Contracting systems were first considered in terms of non-Riemannian Finsler metrics
by Lewis \cite{Lewis1949}.
Recent work on contraction analysis has also made use of non-Riemannian metrics, to identify
system properties \cite{Sontag2010, Aminzare2013}, analyse a wider class of systems
\cite{Russo2013, Fiore2016} and generalize 
contraction ideas to allow different behaviours \cite{Margaliot2016}. 

Forni and Sepulchre have recently provided a framework for contraction analysis that
encompasses approaches using Riemannian metrics and approaches using other metrics
through a generalization to Finsler geometry \cite{Forni2014a}.  Furthermore, they
have taken the first steps towards unifying contraction analysis and traditional
Lyapnuov theory, in an effort to make the powerful tools of Lyapnuov theory available
in contraction analysis.

The primary contribution of this paper is a generalization of CCMs to Finsler manifolds, 
unifying the frameworks of Manchester and Slotine \cite{Manchester2017} and Forni and 
Sepulchre \cite{Forni2014a}.  This removes the restriction that metrics 
are Riemannian, allowing CCM controllers to be applied to a larger class of systems.
Further, we provide a sampled data control construction that does not require the
computation of minimal paths between trajectories.  This controller allows the CCM
method to be applied in cases where either a minimal path does not exist, or it is
too costly to compute a minimal path online.

The remainder of this paper is structured as follows.  In Section~\ref{notation}, we
present the notation used throughout and review some fundamental definitions and
results.  In Section~\ref{contraction}, we provide a differential characterization of
stable closed loop systems.  In Section~\ref{open}, we state and prove the
fundamental result of this paper: a generalization of the open loop controller given
by Manchester and Slotine \cite{Manchester2017} to Finsler manifolds.  We use this open 
loop controller to construct several sampled data, closed loop controllers in 
Section~\ref{closed}. Concluding remarks are made in Section~\ref{conclusions}.

\section{Notation and Preliminaries}\label{notation}

In this section, we introduce the notation used throughout the paper and review
several important results and definitions that underlie the work presented in the
following sections.

We adopt the notation of \cite{Forni2014a} and \cite{Manchester2017}.
A manifold $\mathcal{M}$ is a couple $(\mathcal{M}, \mathcal{A})$ where
$\mathcal{M}$ is a set and $\mathcal{A}$ is a maximal atlas of $\mathcal{M}$ that
induces a Hausdorff, second countable topology.  The tangent space at $x$ and tangent 
bundle for $\mathcal{M}$ are denoted by $\mathcal{T}_x\mathcal{M}$ and
$\mathcal{TM}$ respectively.  We denote by $\mathbb{R}^+$ the set $\{x \in \mathbb{R}
: x \geq 0\}$.

This paper considers control-affine nonlinear dynamical systems defined over a
manifold $\mathcal{M}$ of dimension $n$. These take the form
\begin{IEEEeqnarray}{rCl}
	\dot{x} &=& f(x) + B(x)\,u \label{dynamics}
\end{IEEEeqnarray}
where $f$ is a $C^1$ vector field that maps $x,t \in \mathcal{M} \times
\mathbb{R}$ to vectors in $\mathcal{T}_x\mathcal{M}$, $B$ is a smooth
function and $u \in \mathbb{R}^n$.
A trajectory is a couple $(x, u)$, $x : \mathbb{R}^+ \rightarrow \mathcal{M}$, $u :
\mathbb{R}^+ \rightarrow \mathbb{R}^n$,
such that $x$ and $u$ satisfy (\ref{dynamics}) for all $t \in \mathbb{R}^+$.
Analysis of the differential dynamics, which characterize the linearization of the
system~(\ref{dynamics}) along trajectories, yield the results of this paper.  The
differential dynamics are given by
\begin{IEEEeqnarray}{rCl}
	\dot{\delta_x} &=& A(x, u)\,\delta_x + B(x)\,\delta_u,
	\label{differential}
\end{IEEEeqnarray}
where $A = \dd{f}{x} + \sum_{i=1}^n \dd{b_i}{x} u_i$, $b_i$ represents the
$i^\text{th}$
column of $B$ and $u_i$ represents the $i^\text{th}$ element of $u$.

Several classes of real functions are referred to in the definitions that follow.
A class $\mathcal{K}$ function $\alpha$ is a locally Lipschitz function 
$\alpha: \mathbb{R}^+ \rightarrow \mathbb{R}^+$ which is strictly increasing with
$\alpha(0) = 0$.  A function $\beta: \mathbb{R}^+
\times \mathbb{R}^+ \rightarrow \mathbb{R}^+$ belongs to class $\mathcal{KL}$  if, for all $t \geq 0$,
$\beta(\cdot, t)$ is a class $\mathcal{K}$ function and for all $s \geq 0$, $\beta(s,
\cdot)$ is nonincreasing and $\lim_{t\rightarrow\infty}\beta(s, t) = 0$.

Forni and Sepulchre \cite{Forni2014a} characterize stability of systems on manifold
in the following definition.
\begin{definition}[Incremental stability \cite{Forni2014a}]
	Consider the system $\dot{x} = f(x, t)$ on a manifold 
	$\mathcal{M}$.  Let $d:
	\mathcal{M} \times \mathcal{M} \rightarrow \mathbb{R}^+$ be a continuous
	distance on $\mathcal{M}$. The system is
	\begin{itemize}
		\item \emph{incrementally stable} on $\mathcal{M}$ with 
			respect to $d$
			if there exists a class $\mathcal{K}$ function $\alpha$ such
			that, for any two trajectories $x_1(t), x_2(t) : \mathbb{R}
			\rightarrow \mathcal{M}$, for all $t \geq t_i$,
			\begin{IEEEeqnarray*}{rCl}
				d(x_1(t), x_2(t)) \leq \alpha(d(x_1(0), x_2(0))).
			\end{IEEEeqnarray*}
		\item \emph{incrementally asymptotically stable} on
			$\mathcal{M}$ if it is incrementally stable and, for any two
			trajectories $x_1(t), x_2(t) : \mathbb{R} \rightarrow 
			\mathcal{M}$,
			\begin{IEEEeqnarray*}{rCl}
				\lim_{t\rightarrow \infty} d(x_1(t), x_2(t)) = 0.
			\end{IEEEeqnarray*}
		\item \emph{incrementally exponentially stable} on
			$\mathcal{M}$ if there exists $K \geq 1$ and $\lambda > 0$ such
			that, for any two trajectories $x_1(t), x_2(t) : \mathbb{R}
			\rightarrow \mathcal{M}$, for all $t \geq t_i$,
			\begin{IEEEeqnarray*}{rCl}
				d(x_1(t), x_2(t)) \leq K\,e^{-\lambda(t-t_i)}
				d(x_1(t_i), x_2(t_i)).
			\end{IEEEeqnarray*}
	\end{itemize}
\end{definition}

We extend these definitions to systems with control input, considering whether or not
the system can be made to converge to a particular trajectory.
\begin{definition}[Stabilizability/controllability]
	Consider the system (\ref{dynamics}) on a manifold 
	$\mathcal{M}$.  Let $d:
	\mathcal{M} \times \mathcal{M} \rightarrow \mathbb{R}^+$ be a continuous
	distance on $\mathcal{M}$. A trajectory $x^* : \mathbb{R} \rightarrow
	\mathcal{M}$ of the system (\ref{dynamics}) is
	\begin{itemize}
		\item \emph{controllable (stabilizable)} on $\mathcal{M}$ with 
			respect to $d$
			if there exists an open loop (resp. closed loop) control
			signal and a class $\mathcal{K}$ function $\alpha$ such
			that, for any trajectory $x(t): \mathbb{R}
			\rightarrow \mathcal{M}$, for all $t \geq t_i$,
			\begin{IEEEeqnarray*}{rCl}
				d(x^*(t), x(t)) \leq \alpha(d(x^*, x)).
			\end{IEEEeqnarray*}
		\item \emph{asymptotically controllable
			(stabilizable)} on
			$\mathcal{M}$ if there exists an open loop (resp. closed
			loop) control signal such that it is incrementally stable and, for any
			trajectory $x(t) : \mathbb{R} \rightarrow \mathcal{M}$,
			\begin{IEEEeqnarray*}{rCl}
				\lim_{t\rightarrow \infty} d(x^*(t), x(t)) = 0.
			\end{IEEEeqnarray*}
		\item \emph{exponentially controllable
			(stabilizable)} on
			$\mathcal{M}$ if there exists an open loop (resp. closed
			loop) control signal, $K \geq 1$ and $\lambda > 0$ such
			that, for any trajectory $x^*(t) : \mathbb{R}
			\rightarrow \mathcal{M}$, for all $t \geq t_i$,
			\begin{IEEEeqnarray*}{rCl}
				d(x^*(t), x(t)) \leq K\,e^{-\lambda(t-t_i)}
				d(x^*(t_i), x(t_i)).
			\end{IEEEeqnarray*}
	\end{itemize}
	All three types of controllability (stabilizability) are termed
	\emph{universal} if any choice of trajectory $x^*$ is controllable
	(stabilizable).
\end{definition}

Classical Lyapunov theory and its recent extension to contraction analysis
\cite{Forni2014a} tells us that the existence of a (Finsler-) Lyapunov function is
equivalent to stability of a system. A candidate Finsler-Layapunov function is 
defined as follows.
\begin{definition}[Finsler-Lyapunov function \cite{Forni2014a}]
	Consider a manifold $\mathcal{M}$ and a $C^1$ function
	$V: \mathcal{TM}\rightarrow \mathbb{R}^+; (x, \delta_x) \mapsto V(x,
	\delta_x)$. $V$ is a \emph{Finsler-Lyapunov function for (\ref{dynamics})} if
	there exist $c_1, c_2, p \in \mathbb{R}$, $c_1, c_2 \geq 0$, $p \geq 1$ and a
	Finsler structure $F: \mathcal{TM}\rightarrow \mathbb{R}^+$ such that, for
	all $(x, \delta_x) \in \mathcal{TM}$,
	\begin{IEEEeqnarray}
		cc_1\,F(x,\delta_x)^p \leq V(x, \delta_x) \leq c_2\,F(x, \delta_x)^p.
		\label{bounds}
	\end{IEEEeqnarray}
	$F$ satisfies:
	\begin{enumerate}
		\item $F$ is a $C^1$ function for every $(x, \delta_x) \in
			\mathcal{TM}$ such that $\delta_x \neq 0$;
		\item $F(x, \delta_x) > 0$ for each $(x, \delta_x) \in \mathcal{TM}$
			such that $\delta_x \neq 0$;
		\item $F(x, \lambda\delta_x) = \lambda\,F(x, \delta_x)$ for every
			$\lambda \geq 0$ and every $(x, \delta_x) \in \mathcal{TM}$;
		\item $F(x, \delta_{x1} + \delta_{x2}) < F(x, \delta_{x1}) + F(x,
			\delta_{x2})$ for every $(x, \delta_{x1}),\;(x, \delta_{x2})
			\in \mathcal{TM}$ such that $\delta_{x1} \neq
			\lambda\delta_{x2}$ for any $\lambda \in \mathbb{R}$.
	\end{enumerate}
	We refer to a manifold $\mathcal{M}$ endowed with a Finsler structure $F$ as
	a \emph{Finsler manifold}.
\end{definition}

A key result of Forni and Sepulchre \cite{Forni2014a} is that, for a system with no 
control input, if a 
candidate Finsler-Lyapunov function $V$ can be found such that
\begin{IEEEeqnarray*}{rCl}
	\dot{V} = \dd{V}{x}\,\dot{x} + \dd{V}{\delta_x}\,\dot{\delta_x} &\leq&
	-\alpha(V),
\end{IEEEeqnarray*}
then the system is incrementally stable, with the type of stability depending on the
form of $\alpha$.  In this case, $V$ is called a \emph{contraction measure}.

The Finsler structure $F$ endows the manifold $\mathcal{M}$ with a global measure of
distance.  Before defining this distance, we introduce some necessary notation.
A curve $\gamma$ on a manifold $\mathcal{M}$ is a function $\gamma : I \subset
\mathbb{R} \rightarrow \mathcal{M}$.  We denote $\partial \gamma/\partial s$ by
$\gamma_s$. A curve is regular if $\gamma_s \neq 0$
for all $s$.  The space $\Gamma(x_1, x_2)$ is defined as the set of all curves
$\gamma: [0, 1] \rightarrow \mathcal{M}$ such that $\gamma(0) = x_1$ and $\gamma(1) =
x_2$.

\begin{definition}[Finsler distance]
	Given a candidate Finsler-Lyapunov function $V$ with Finsler structure $F$
	and defining $I = [0, 1]$, the distance $d : \mathcal{M} \times \mathcal{M} \rightarrow \mathbb{R}^+$ is
	given by
	\begin{IEEEeqnarray}{rCl}
		d(x_1, x_2) = \inf_{\Gamma(x_1, x_2)} \int_I F(\gamma(s),
		\gamma_s(s))\mathrm{d}s. \label{finsler_distance}
	\end{IEEEeqnarray}
\end{definition}

Note that in general, $d$ is not symmetric, that is, $d(x_1, x_2) \neq d(x_2, x_1)$.
$d$ is however positive definite and satisfies the triangle inequality.  
If $\int_I F(\gamma(s), \dot\gamma(s))\mathrm{d}s = d(x_1, x_2)$, $\gamma$ is a
minimising geodesic. The Finsler manifold $\mathcal{M}$ is said to be forward
geodesically complete if every geodesic $\gamma(s)$ defined on $s \in [a, b)$ can be
extended to a geodesic defined on $s \in [a, \infty)$.  The Hopf-Rinow Theorem states
that any two points $x_1, x_2 \in \mathcal{M}$ can be connected by a minimising
geodesic if $\mathcal{M}$ is forward geodesically complete.  We refer the
reader to Bao et al. \cite{Bao2000} for a full treatment of the Hopf-Rinow theorem and geodesics
on Finsler manifolds.

For notational convenience, this paper only considers time invariant $f$, $B$ and
$V$.  However, the results extend in a straightforward manner to the time varying
case.

\section{Contraction of Systems with Control Inputs}\label{contraction}

We begin in a similar manner to \cite[Prop. 1]{Manchester2017}, by examining 
conditions on a contraction measure $V$ for a system with
control inputs, when a controller has been found that makes the system contract while
allowing all solutions of the dynamics (\ref{dynamics}) to remain feasible.

\begin{proposition}\label{PROP1}
	Suppose that, for the system (\ref{dynamics}) on a smooth manifold
	$\mathcal{M}$, there exists a smooth feedback control of the form $u = k(x,
	t) + v$ such that there exists a candidate Finsler-Lyapunov
	function $V$ (with Finsler structure $F$) that gives
	\begin{IEEEeqnarray}{rCl}
		\dot{V} = \dd{V}{x}\,\dot{x} + \dd{V}{\delta_x}\,\dot{\delta_x} &\leq&
		-\alpha(V) \label{contract}
	\end{IEEEeqnarray}
	for every $t\in\mathbb{R}$, $x\in\mathcal{M}$, $\delta_x \in
	\mathcal{T}_x\mathcal{M}$ and $v \in \mathbb{R}^n$ and
	some $\alpha:\mathbb{R}^+\rightarrow\mathbb{R}^+$. Then, for all $\delta_x
	\neq 0$,
	\begin{IEEEeqnarray}{rCl}
		\dd{V}{\delta_x}B &=& 0 \implies \dd{V}{x}\left(f + Bu\right) +
		\dd{V}{\delta_x}A\delta_x \leq -\alpha(V). \label{prop1}
	\end{IEEEeqnarray}
\end{proposition}
~\\
\begin{proof}
	We first note that $\delta_u = \dd{u}{x} \delta_x = K \delta_x$.
	Substituting the dynamics (\ref{dynamics}) and differential dynamics
	(\ref{differential}) in (\ref{contract}) gives
	\begin{IEEEeqnarray*}{rCl}
		\dd{V}{x}(f + Bu) + \dd{V}{\delta_x}(A\delta_x +
		B\delta_u) \leq -\alpha(V)\\
		\dd{V}{x}f + \dd{V}{x}B(k + v) + \dd{V}{\delta_x}\left(\dd{f}{x}
		+ \sum^m_{i=1} \dd{b_i}{x}\left(k_i + v_i\right)\right)\delta_x\\ +
		\dd{V}{\delta_x}B K \delta_x \leq -\alpha(V).
	\end{IEEEeqnarray*}
	As this is affine in $v$, for the right hand side to remain bounded for
	unbounded $v$, we require
	\begin{IEEEeqnarray}{rCl}
		\dd{V}{x}\,b_i + \dd{V}{\delta_x}\,\dd{b_i}{x} &=& 0\label{bzero}
	\end{IEEEeqnarray}
	For all $i$, where $b_i$ is the $i^\text{th}$ column of $B$. This gives
	\begin{IEEEeqnarray*}{rCl}
		\dd{V}{x}\,f + \dd{V}{\delta_x}\,\dd{f}{x}\,\delta_x +
		\dd{V}{\delta_x}\,B\,K\,\delta_x \leq -\alpha(V).
	\end{IEEEeqnarray*}
	The result follows from letting $(\partial V/\partial \delta_x)\,B = 0$ and
	adding (\ref{bzero}).	
\end{proof}

\section{Open Loop Control Synthesis}\label{open}

The result of Proposition~\ref{PROP1} leads to the question of whether
Condition~(\ref{prop1}) implies the existence of a 
stabilizing control law. In this section, we show that if a  contraction measure 
$V$ can be found such that 
(\ref{prop1}) is true, the system can be universally stabilized by an open loop 
control signal.  This provides a generalization of the open loop results of 
\cite[Th. 1]{Manchester2017}.

\begin{theorem}\label{TH1}
	Consider the system (\ref{dynamics}), (\ref{differential}) on a smooth
	manifold $\mathcal{M}$ with $f \in C^2$.  Suppose there exists a
	candidate Finsler-Lyapunov function $V \in C^\infty$ with 
	Finsler structure $F$ such that
	\begin{IEEEeqnarray}{rCl}
		\dd{V}{\delta_x}B &=& 0 \implies \dd{V}{x}\left(f + Bu\right) +
		\dd{V}{\delta_x}A\delta_x < -\alpha(V) \label{th1}
	\end{IEEEeqnarray}
	for every $t\in\mathbb{R}^+$, $x\in\mathcal{M}$, $u\in\mathbb{R}^n$ and 
	$\delta_x \in
	\mathcal{T}_x\mathcal{M},\; \delta_x \neq 0$.  Furthermore, suppose that for
	all compact subsets $X \subset \mathbb{R}^n$, for all $x\in X$, and for all
	compact subsets $Y \subset \mathbb{R}^n$ not containing 0, for all $\delta_x
	\in Y$, the ratio
	\begin{IEEEeqnarray}{C}
		\frac{\dd{\dot{V}}{u}}{\dd{V}{\delta_x} B B\tran
		\dd{V}{\delta_x}\tran} \label{ratio}
	\end{IEEEeqnarray}
	is bounded, where
	\begin{IEEEeqnarray*}{rCl}
		\dd{\dot{V}}{u} &=& \dd{V}{x} B + \dd{V}{\delta_x} \dd{B}{x}
		\delta_x.
	\end{IEEEeqnarray*}
	Then there exists an open loop control law such that the system is
	\begin{itemize}
		\item universally controllable on $\mathcal{M}$ if $\alpha(s) = 0$ for all
			$s \geq 0$;
		\item universally asymptotically controllable on $\mathcal{M}$ if 
			$\alpha$ is a class $\mathcal{K}$ function;
		\item universally exponentially controllable on $\mathcal{M}$ if 
			$\alpha(s) = \lambda\,s > 0$ for each $s>0$.
	\end{itemize}
	We refer to the function $V$ as a \emph{(Finsler) control contraction metric}
	(CCM).
\end{theorem}
The condition that (\ref{ratio}) is bounded is true for any system meeting the
conditions of Proposition~\ref{PROP1} or the strong
conditions given in \cite[Sec. III. A.]{Manchester2017}. We now give the control law 
construction and then prove stabilizability.

\subsection{Open Loop Control Construction}\label{OLcontrol}

We construct a local control law that stabilizes the differential dynamics
(\ref{differential}), following the construction of \cite[Lemma 2]{Manchester2017}.
We then integrate this along a path connecting an arbitrary target trajectory with
the current trajectory and apply the control signal corresponding to the integral
evaluated at the current trajectory.

Define
\begin{IEEEeqnarray*}{rCl}
	a(x, \delta_x, u) &=& \dd{V}{x}\,(f + Bu) + \dd{V}{\delta_x}\,A\delta_x +
	\alpha(V)\\
	b(x, \delta_x) &=& \dd{V}{\delta_x}\,B\,B\tran\dd{V}{\delta_x}\tran.
\end{IEEEeqnarray*}
Let 
\begin{IEEEeqnarray*}{rCl}
	\rho(x, \delta_x, u) &=& \left\{ \begin{array}{c l}
						0 & \text{if}\; a < 0\\
						\frac{a + \sqrt{a^2 + b^2}}{b} &
						\text{otherwise.}
	\end{array}\right.
\end{IEEEeqnarray*}
The differential feedback control is given by
\begin{IEEEeqnarray}{rCl}
	k_\delta(x, \delta_x, u) &=& -\rho(x, \delta_x, u)\,B(x)\tran\dd{V(x,
	\delta_x)}{\delta_x}\tran.\label{kd}
\end{IEEEeqnarray}

The open loop control signal to stabilize the system to a target trajectory $(x^*(t),
u^*(t))\in \mathcal{C}\times \mathbb{R}^n$ in a time interval $\mathcal{T}$ of the
form $[t_i, t_e)$, $[t_i,
t_e]$, or $[t_i, \infty)$ is then calculated as follows:
\begin{enumerate}
	\item Measure $x(t_i)$ and construct a smooth path $c(t_i, s) \in
		\Gamma(x^*(t_i), x(t_i))$.
	\item \label{control} Solve the following equation for $k_p$:
		\begin{IEEEeqnarray*}{R}
			k_p(c, u^*, t, s) = u^*\; + \\
			 \int^s_0 k_\delta(c(t, s), c_s(t, s),
			k_p(c, u^*, t, s), t)\;\mathrm{d}s,
		\end{IEEEeqnarray*}
		where $c_s = \partial c/\partial s$.
	\item For each $t\in\mathcal{T}$, apply the control signal $u(t) = k_p(c(t, s),
		u^*(t), t, 1)$, where $c(t, s)$ is the forward image of $c(t_i, s)$ with
		the path of controls defined in Equation~(\ref{control}).  
		That is, for all
		$s\in[0,1]$ and $t\in\mathcal{T}$, $c(t,s)$ is a solution to
		\begin{IEEEeqnarray*}{rCl}
			\frac{\mathrm{d}}{\mathrm{d}t}\,c(t,s) &=& f(c(t,s), t)\; +\\
			&& B(c(t,s),t)\,k_p(c(t), u^*(t), t, s).
		\end{IEEEeqnarray*}
\end{enumerate}

\begin{proof}[Proof of Theorem \ref{TH1}]
	It follows from \cite[Lemma 2]{Manchester2017} that the differential 
	control (\ref{kd}) makes the extended system
	dissipative with respect to the storage function $V$ and supply rate
	$\alpha(V)$, that is,
	\begin{IEEEeqnarray}{rCl}
		\dot{V} &=& \dd{V}{x}(f + B\,u) + \dd{V}{\delta_x}(A\delta_x +
		Bk_\delta) < -\alpha(V).\label{dissipative}
	\end{IEEEeqnarray}
	Indeed, substituting (\ref{kd}) into the left hand side of (\ref{dissipative}) gives
	\begin{IEEEeqnarray*}{rCl}
		\dot{V} &=& a - \rho\,b - \alpha(V)\\
		&=& -\alpha(V) - \sqrt{a^2 + b^2}\\
		&<& -\alpha(V).
	\end{IEEEeqnarray*}

	We now show that the differential control signal (\ref{kd}) is integrable
	along regular curves in $\mathcal{M}$.  That is, for any regular curve
	$c:[0,1] \rightarrow \mathcal{M}$ and any $u_0\in \mathbb{R}^n$,
	$t\in\mathbb{R}^+$, a unique solution of the following integral equation
	exists on $s\in [0,1]$:
	\begin{IEEEeqnarray}{rCl}
		v(s) &=& u_0 + \int^s_0 k_\delta\left(c(s), c_s(s), v(s)
		\right)\mathrm{d}s.\label{int}
	\end{IEEEeqnarray}

	The condition (\ref{th1}) implies either $b > 0$ or $a < 0$.  It follows from 
	\cite[Th. 1]{Sontag1989} that $\rho$ is smooth for all $x$, $u$ and $\delta_x
	\neq 0$, and the apparent discontinuity at $b = 0$ is removed by setting
	$\rho = 0$ when $b = 0$.  Smoothness of $\rho$ implies smoothness of
	$k_\delta$ when $\delta_x \neq 0$, which is the case in Equation~(\ref{int})
	where $\delta_x$ is set to $c_s(s)$, which is non-zero by regularity of $c(s)$.

	It follows from \cite[Th. 3.2]{Khalil2002} that a unique solution to (\ref{int})
	exists if $k_\delta$ is a globally Lipschitz function with respect to its
	third argument for $s\in[0, 1]$.  As $B$ and $V$ are continuously
	differentiable and smooth respectively, the product $B(\partial V/\partial
	\delta_x)$ is bounded on
	closed intervals. Hence if $\rho$ is Lipschitz with respect to $u$ on
	$s\in[0,1]$, so too is $k_\delta$. As $\rho$ is smooth, it is globally
	Lipschitz if its derivative with respect to its third argument is bounded.  
	This is clear for $b \leq 0$.  For $b > 0$, noting that the only dependence
	$\rho$ has on $u$ is via $a$, we have
	\begin{IEEEeqnarray*}{rCl}
		\dd{\rho}{u} &=& \dd{}{u}\left(\frac{a(u) + \sqrt{a^2(u) +
		b^2}}{b}\right)\\
		&=& \frac{1}{b}\dd{a}{u} \left(1 + \frac{a(u)}{\sqrt{a^2(u) +
		b^2}}\right).
	\end{IEEEeqnarray*}
	Since $a = 0 \implies b > 0$ and $a$ is affine in $u$, the only term that can
	be unbounded is $(1/b)(\partial a/\partial u)$.  However, this is precisely
	the term which is bounded by Condition~(\ref{ratio}).  Hence $\rho$ is
	Lipschitz with respect to $u$ for $s\in[0, 1]$ and a solution to (\ref{int})
	exists.

	We now show that applying the control law of Section~\ref{OLcontrol} makes
	the initial trajectory $x$ converge to the (arbitrary) chosen 
	trajectory $x^*$. Universal exponential stabilizability follows. 

	Consider a regular curve $c(t, s) \in \Gamma(x^*(t_i), x(t_i))$.  Then, for all $t
	\geq t_i$, we have $c(t, 0) = x^*(t)$ and $c(t, 1) = x(t)$.  Furthermore, for
	all $t \geq t_i$ and all $s \in [0,1]$, $c_s = \partial c/ \partial s$
	satisfies the differential dynamics (\ref{differential}):
	\begin{IEEEeqnarray*}{rCl}
		\frac{\mathrm{d}}{\mathrm{d}t} c_s(t,s) &=& A\,c_s + B\,k_\delta.\\
	\end{IEEEeqnarray*} 
	It follows from (\ref{dissipative}) that
	\begin{IEEEeqnarray}{rCl}
		\frac{\mathrm{d}}{\mathrm{d}t} V(c(t), c_s(t)) &<& -\alpha(V(c(t),
		c_s(t))). \label{diss1}
	\end{IEEEeqnarray}

	We now consider three cases, corresponding to the three forms of $\alpha$
	given in the statement of Theorem~\ref{TH1}.

	If $\alpha(s) = 0$, (\ref{diss1}) gives $\frac{\mathrm{d}}{\mathrm{d}t} V(c(t), c_s(t))
	< 0$, so $V(c(t), c_s(t)) < V(c(t_i), c_s(t_i))$ for all $t \geq t_i$.
	As $V$ is non-negative, this gives $V(c(t), c_s(t))^{1/p} < V(c(t_i),
	c_s(t_i))^{1/p}$. It follows that
	\begin{IEEEeqnarray*}{rCl}
		d(x^*(t), x(t)) &\leq& \int_I F(c(t), c_s(t)) \mathrm{d}s\\
		&\leq& c_1^{-\frac{1}{p}} \int_I V(c(t),
		c_s(t))^{\frac{1}{p}}\mathrm{d}s\\
		&<& c_1^{-\frac{1}{p}}\int_I V(c(t_i),
		c_s(t_i))^{\frac{1}{p}}\mathrm{d}s\\
		&\leq& \left(\frac{c_2}{c_1}\right)^{\frac{1}{p}}
		\int_I F(c(t_i), c_s(t_i)) \mathrm{d}s.
	\end{IEEEeqnarray*}
	As the choice of $x^*$ is arbitrary, this implies that the system is
	universally controllable.

	If $\alpha(V)$ is a class $\mathcal{K}$ function, (\ref{diss1}) again gives 
	$\frac{\mathrm{d}}{\mathrm{d}t} 
	V(c(t), c_s(t)) < 0$, so the system is universally controllable.
	Furthermore, as shown in the proof of \cite[Th. 1]{Forni2014a}, there exists a
	$\mathcal{KL}$ function $\beta$ such that
	\begin{IEEEeqnarray*}{rCl}
		V(c(t), c_s(t)) &\leq& \beta(C(c(t_i), c_s(t_i), t-t_i)).
	\end{IEEEeqnarray*}
	Integrating with respect to $s$,
	\begin{IEEEeqnarray*}{rCl}
		\int_I F(c(t), c_s(t))^p \mathrm{d}s &\leq& c_1^{-1} \int_I \beta(V(c(t_i),
		c_s(t_i)), t-t_i)\mathrm{d}s\\
		d(x^*(t), x(t)) &\leq& c_1^{-\frac{1}{p}} \int_I \beta( V(c(t_i),
		c_s(t_i)), t-t_i)^{\frac{1}{p}}\\
		\lim_{t\rightarrow \infty} d(x^*(t), x(t)) &\leq& c_1^{-\frac{1}{p}}
		\\ &&
		\lim_{t\rightarrow \infty} \int_I \beta( V(c(t_i), c_s(t_i)), 
		t-t_i)^{\frac{1}{p}}\\
		&=& 0,
	\end{IEEEeqnarray*}
	where the final equality follows from the definition of a $\mathcal{KL}$
	function and Lebesgue's dominated convergence theorem.  As the choice of
	$x^*$ is arbitrary, universal asymptotic controllability follows.
	
	If $\alpha(V) = \lambda\,V$, (\ref{diss1}) and \cite[Th. 6.1]{hale1969} give
	\begin{IEEEeqnarray*}{rCl}
		V(c(t), c_s(t)) &<& e^{-\lambda{(t-t_i)}} V(c(t_i), c_s(t_i))\\
		V(c(t), c_s(t))^{\frac{1}{p}} &<& e^{\frac{-\lambda}{p}{(t-t_i)}}
		V(c(t_i), c_s(t_i))^{\frac{1}{p}}.
		\IEEEyesnumber \label{alphalambda}
	\end{IEEEeqnarray*}
	We then have
	\begin{IEEEeqnarray*}{rCl}
		d(x^*(t), x(t)) &\leq& \int_I F(c(t), c_s(t))\mathrm{d}s\\
		&\leq& c_1^{-\frac{1}{p}} \int_I V(c(t), c_s(t))^\frac{1}{p}
		\mathrm{d}s\\
		&\leq& c_1^{-\frac{1}{p}} \int_I e^{\frac{-\lambda}{p}{(t-t_i)}} 
		V(c(t_i), c_s(t_i))^\frac{1}{p} \mathrm{d}s\\
		&\leq&  \left(\frac{c_2}{c_1}\right)^{\frac{1}{p}}
		e^{\frac{-\lambda}{p}{(t-t_i)}} 
		\int_I F(c(t_i), c_s(t_i)) \mathrm{d}s
	\end{IEEEeqnarray*}
	
	That is, the trajectories $x$ and $x^*$ converge exponentially with rate
	$\lambda/p$.
	Furthermore, if $c(t_i)$ is a minimising geodesic,
	\begin{IEEEeqnarray*}{rCl}
		d(x^*(t), x(t)) &<& 
		\left(\frac{c_2}{c_1}\right)^{\frac{1}{p}}
		e^{-\frac{\lambda}{p}(t-t_i)} \\ && d(x^*(t_i), x(t_i)).
	\end{IEEEeqnarray*}
	This implies that, if $c$ is a minimising geodesic, the overshoot is bounded 
	above by
	$(c_2/c_1)^{(1/p)}$.
	As the choice of $x^*$ is arbitrary, this proves Theorem~\ref{TH1} for the
	case $\alpha(V) = \lambda\,V$.
\end{proof}

\begin{remark}
	The control scheme proposed in Section~\ref{OLcontrol} appears difficult to
	compute.  Computation of this scheme is dealt with (under the restricted class 
	of Riemannian metrics) in several other papers. Manchester and Slotine 
	\cite{Manchester2017} define a continuous feedback control
	which removes the need for Step 3 of the open loop construction. Leung and
	Manchester \cite{Leung2017} present a pseudospectral approach for the
	computation of the path in Step 1, and show that this approach is more
	efficient than nonlinear Model Predictive Control.  While we do not detail
	any computational methods, the following two examples illustrate the construction
	of open loop controllers for simple systems.
\end{remark}
	
\begin{example}\label{example1}
	Let $\mathcal{M} = \mathbb{R}^2$ and consider the system
	\begin{IEEEeqnarray*}{rCl}
		\dot{x} &=& \left(\begin{array}{c c} 1 & 0\\ 0 & -1\end{array}\right)x
			+ \left(\begin{array}{c} 1\\0\end{array}\right) u.
	\end{IEEEeqnarray*}
	Let $V = \delta_1^4 + \delta_2^4$, with one possible Finsler structure given
	by $V^{1/4}$.  Then
	\begin{IEEEeqnarray*}{rCl}
		\dd{V}{\delta_x} B &=& 4\delta_1^3,
	\end{IEEEeqnarray*}
	which is zero at $\delta_1 = 0$, and
	\begin{IEEEeqnarray*}{rCl}
		\dd{V}{\delta_x}\left(\dd{f}{x} + \dd{B}{x} u\right)\delta_x &=&
		4\left(\delta_1^4 - \delta_2^4\right),	
	\end{IEEEeqnarray*}
	which is strictly less than $(-4+\epsilon)\delta_2^4$ when $\delta_1 = 0$.
	Furthermore, the numerator of (\ref{ratio}) is always zero as $\partial
	B/\partial x = 0$, so the ratio (\ref{ratio}) is bounded.  Hence it follows
	from Theorem~\ref{TH1} that this system is universally controllable.

	Treating $\epsilon$ as zero, we have
	\begin{IEEEeqnarray*}{rCl}
		\rho &=& \frac{-1 -\sqrt{1 + 16\delta_1}}{4 \delta_1^2}\\
		k_\delta &=& -\delta_1 -\delta_1\sqrt{1 + 16 \delta_1^4}.
	\end{IEEEeqnarray*}

	Now suppose that our initial position is $(1, 1)$ and our desired trajectory
	is $x^*(t) = (0, 0), u^* = 0$. A path connecting our initial and desired
	position is given by $c(t_i, s) = (s, s)$.  We then have the following
	equation for $k_p$:
	\begin{IEEEeqnarray}{rCl}
		k_p &=& \int_0^s -\dd{c}{s} -\dd{c}{s}\sqrt{1 + 16
		\left(\dd{c}{s}\right)^4}\;\mathrm{d}s\label{exampleControl}.
	\end{IEEEeqnarray}
	This is solved approximately by discretising both in time and with respect to $s$ 
	along the curve $c(t, s)$. At each time step, Equation~(\ref{exampleControl})
	is solved approximately by quadrature.  The forward image of each discretised
	point on $c(t, s)$ is then calculated by numerical integration of the system
	dynamics with the newly computed control signal.

	Figure~\ref{example1fig} illustrates the time response of the unstable state 
	given this control scheme.  The response of the same system with control
	calculated using the Riemannian metric $V = \delta_1^2 + \delta_2^2$ is also illustrated.

	\begin{figure}
		\vspace{0.2cm}
		\centering
		\newlength\figureheight
		\newlength\figurewidth
		\setlength\figureheight{0.18\textwidth}
		\setlength\figurewidth{0.4\textwidth}
		\input{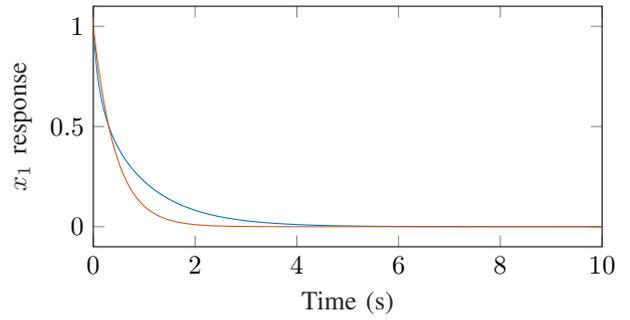}
		\caption{Response of state $x_1$ of the system in
		Example~\ref{example1} under control computed with $V = \delta_1^2 +
		\delta_2^2$ (orange) and $V = \delta_1^4 + \delta_2^4$ (blue).}
		\label{example1fig}
	\end{figure}

\end{example}

\begin{example}\label{example2}
	Consider the one dimensional system $\dot{\theta} = -\sin\theta + u$, which
	approximates an overdamped pendulum.  Let $\mathcal{M} = [0, \pi]$. Note that 
	the system has a stable equilibrium at $\theta = 0$ (the downright
	equilibrium) and an unstable equilibrium at $\theta = \pi$ (the upright
	equilibrium).  We compare control computed 
	with a Riemannian (and hence symmetric) metric $V_1 = 4\delta_\theta^2$, and 
	an asymmetric Finsler metric $V_2 = (2\sqrt{\delta_\theta^2} - \delta_\theta)^2$ 
	(which is the square of a Randers metric \cite[Sec. 1.3C]{Bao2000}).  Both 
	metrics satisfy the conditions of Theorem~\ref{TH1} with a Finsler structure 
	given by $F_i = \sqrt{V_i}$, except that $V_2$ is only $C^1$ and not smooth.  
	This means the proof of integrability does not apply.  However, we find in 
	this case that a controller can be computed.  
	
	The
	induced Finsler distance $d(x_1, x_2)$ may be thought of as the cost to go 
	from state $x_1$ to state $x_2$.  If $x_1 > x_2$, use of the asymmetric $V_2$
	and $F_2$
	gives $d(x_1, x_2) > d(x_2, x_1)$. In real terms, it is more costly to rotate the pendulum
	downwards than upwards.  Intuitively, there is no reason for a
	``cost to go'' to be symmetric - in this case, rotating the pendulum
	downwards represents a loss of potential energy and can be deemed as more
	expensive than the corresponding gain in potential energy. The controller 
	computed with this asymmetric metric
	uses larger control input to move the pendulum from $\theta = 0$ to 
	$\theta = \pi$ (Figure~\ref{example2fig1}) than in the opposite direction
	(Figure~\ref{example2fig2}).  Comparing Figures~\ref{example2fig1} 
	and~\ref{example2fig2}, we see that the asymmetric controller regulates
	faster than the symmetric controller when moving the pendulum from $0$
	(pointing downwards) to $\pi$ (pointing upwards),
	while the symmetric controller regulates faster moving in the opposite
	direction.  Peak control input from the symmetric controller is roughly equal
	in each case, while peak control input from the asymmetric controller is six
	times larger when moving from $0$ to $\pi$.

	\begin{figure}
		\vspace{0.2cm}
		\centering
		\setlength\figureheight{0.4\textwidth}
		\setlength\figurewidth{0.4\textwidth}
		\input{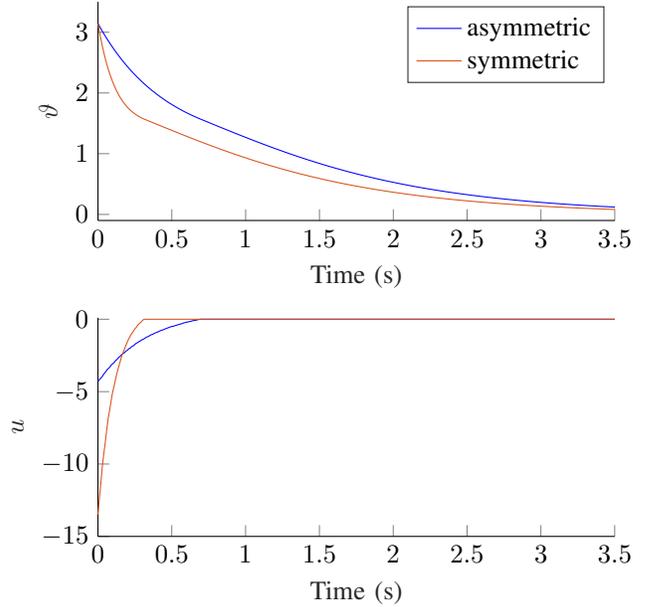}
		\caption{Response of the pendulum in Example~\ref{example2} moving
		from $\pi$ to 0 under control computed with  $V =
		(2\sqrt{\delta_\theta^2} - \delta_\theta)^2$ (asymmetric) and $V =
		4\delta_\theta^2$ (symmetric).}
		\label{example2fig1}
	\end{figure}

	\begin{figure}
		\vspace{0.2cm}
		\centering
		\setlength\figureheight{0.4\textwidth}
		\setlength\figurewidth{0.4\textwidth}
		\input{example2fig2.tikz}
		\caption{Response of the pendulum in Example~\ref{example2} moving
		from $0$ to $\pi$ under control computed with  $V =
		(2\sqrt{\delta_\theta^2} - \delta_\theta)^2$ (asymmetric) and $V =
		4\delta_\theta^2$ (symmetric).}
		\label{example2fig2}
	\end{figure}
\end{example}

\begin{example}
	Let $\mathcal{M} = \mathbb{R}$ and consider the system
	\begin{IEEEeqnarray*}{rCl}
		\dot{x} = -x + x^2 u.
	\end{IEEEeqnarray*}
	This example illustrates the importance of the ratio (\ref{ratio}) remaining
	bounded.  With a Lyapunov function of $V = \delta_x^2$ and Finsler structure
	$F = \sqrt{V}$, we have
	\begin{IEEEeqnarray*}{rCl}
		\dd{V}{\delta_x} B &=& 2\delta_x x^2,\; \text{and}\\
		\dd{V}{\delta_x} \left( \dd{f}{x} + \dd{B}{x}u\right)\delta_x &=&
		-2\delta_x^2 + 4\delta_x^2 u x.
	\end{IEEEeqnarray*}
	Hence the system and Lyapunov function meet the requirement~(\ref{th1}).
	However, computing the ratio~(\ref{ratio}) gives $1/x^3$, which is unbounded
	as $x \rightarrow 0$.  This means the control signal cannot be integrated.
	Intuitively, if we begin with an initial condition of zero, we are at a
	stable equilibrium with no control input and can never leave.
\end{example}

\section{Closed Loop Control Synthesis}\label{closed}

A natural question is whether the open loop results of the previous section can be
adapted to a closed loop controller.  In this section, we develop sampled data
feedback controllers that guarantee universal stabilizability.  We first give a
general sampled data control construction, then prove its properties.

\subsection{Closed Loop Control Construction}\label{CLlambda}

Given a continuous time period $\mathcal{T} = [t_i, t_e)$, $[t_i, t_e]$ or $[t_i,
\infty)$, we construct a closed loop control calculated at discrete times in
$\mathcal{T}$ as follows.
\begin{enumerate}
	\item At the initial time $t_i$, measure the present state $x(t_i)$ and
		construct a path $c_i \in \Gamma(x^*(t_i), x(t_i))$.
	\item Run the open loop control constructed in Section~\ref{OLcontrol} 
		on the interval
		$[t_i, t_{i+1})$ for some $t_{i+1} > t_i$. \label{steptwo}
	\item At time $t_{i+1}$, compute a new path $c_{i+1}$ such that 
		\begin{IEEEeqnarray*}{C}
			\int_I V(c_{i+1}(t_{i+1}), c_{s,
			i+1}(t_{i+1}))^{\frac{1}{p}}\mathrm{d}s\\
			\leq \int_I V(c_{i}(t_{i+1}), c_{s,
			i}(t_{i+1}))^{\frac{1}{p}} \mathrm{d}s \IEEEyesnumber \label{newpath}
		\end{IEEEeqnarray*}
		for all $s\in I$ and return to step~\ref{steptwo}.
\end{enumerate}
Note that the sample times $t_i, t_{i+1}, \ldots$ may be chosen arbitrarily.


\begin{proposition}\label{CLprop}
	Consider the system (\ref{dynamics}), (\ref{differential}) on a smooth
	manifold $\mathcal{M}$ with $f \in C^2$.  Suppose there exists a
	Finsler-Lyapunov function $V \in C^\infty$ with Finsler structure $F$ that
	satisfies the conditions of Theorem~\ref{TH1}.  Then the system is
	\begin{itemize}
		\item universally stabilizable on $\mathcal{M}$ via the control
			construction of Section~\ref{CLlambda} if $\alpha = 0$. 
		\item universally exponentially stabilizable on $\mathcal{M}$ via the
			control construction of Section~\ref{CLlambda} 
			if $\alpha(V) = -\lambda V$ for fixed $\lambda > 0$.
	\end{itemize}
\end{proposition}

\begin{remark}
	In the open loop case, if $\alpha$ is a class $\mathcal{K}$ function, the
	system exhibits asymptotic stabilizability.  However, as we have no information
	about the rate of convergence, for the sampled data
	controller we can only guarantee regular stabilizability (the distance
	between the current and target trajectories remains bounded for all time).
\end{remark}

\begin{proof}[Proof of Proposition~\ref{CLprop}]
	First consider the case $\alpha = 0$. Then, for the first time period,
	(\ref{diss1}) gives 
	\begin{IEEEeqnarray*}{rCl}
		\int_I V(c(t), c_s(t))^{\frac{1}{p}}\mathrm{d}s &\leq&
		\int_I V(c_i(t_i), c_{s, i}(t_i))^{\frac{1}{p}}\mathrm{d}s.\\
	\end{IEEEeqnarray*}
	On the second time period, we have
	\begin{IEEEeqnarray*}{l}
		\int_I V(c_{i+1}(t), c_{s, i+1}(t))^{\frac{1}{p}}\mathrm{d}s\\
		\leq \int_I 
		V(c_{i+1}(t_{i+1}), c_{s, i+1}(t_{i+1}))^{\frac{1}{p}}\mathrm{d}s\\
		\leq \int_I 
		V(c_{i}(t_{i+1}), c_{s, i}(t_{i+1}))^{\frac{1}{p}}\mathrm{d}s\\
		\leq \int_I 
		V(c_{i}(t_{i}), c_{s, i}(t_{i}))^{\frac{1}{p}}\mathrm{d}s.\\
	\end{IEEEeqnarray*}
	By induction, it follows that, on any time period,
	\begin{IEEEeqnarray*}{rCl}
		\int_I V(c_{i+j}(t), c_{s, i+j}(t))^{\frac{1}{p}}\mathrm{d}s &\leq&  
		\int_I V(c_{i}(t_{i}), c_{s, i}(t_{i}))^{\frac{1}{p}}\mathrm{d}s.\\
	\end{IEEEeqnarray*}
	This gives, for all $t$,
	\begin{IEEEeqnarray*}{rCl}
		d(x^*(t), x(t)) &\leq& \int_I F(c_{i + j}(t), c_{s, i+j}(t))\mathrm{d}s\\
		&\leq& c_1^{-\frac{1}{p}} \int_I V(c_{i+j}(t), c_{s, i+j}(t))^\frac{1}{p}\mathrm{d}s\\
		&<&  c_1^{-\frac{1}{p}} \int_I V(c_{i}(t_{i}), c_{s, i}(t_{i}))^{\frac{1}{p}}\mathrm{d}s\\
		&\leq& \left(\frac{c_2}{c_1}\right)^{-p} \int_I F(c_i(t_i), c_{s, i}(t_i))\mathrm{d}s.\\
	\end{IEEEeqnarray*}
	This proves universal stabilizability for the case $\alpha = 0$.

	Now consider the case $\alpha(V) = -\lambda V$. On the first time interval, 
	(\ref{alphalambda}) gives 
	\begin{IEEEeqnarray*}{rCl}
		\int_I V(c(t), c_s(t))^{\frac{1}{p}}\mathrm{d}s &<&
		e^{-\frac{\lambda}{p}(t-t_i)} \int_I 
		V(c_i(t_i), c_{s, i}(t_i))^{\frac{1}{p}}\mathrm{d}s.\\
	\end{IEEEeqnarray*}
	On the second time period, we have
	\begin{IEEEeqnarray*}{l}
		\int_I V(c_{i+1}(t), c_{s, i+1}(t))^{\frac{1}{p}}\mathrm{d}s\\
		< e^{-\frac{\lambda}{p}(t-t_{i+1})} \int_I 
		V(c_{i+1}(t_{i+1}), c_{s, i+1}(t_{i+1}))^{\frac{1}{p}}\mathrm{d}s\\
		\leq e^{-\frac{\lambda}{p}(t-t_{i+1})} \int_I 
		V(c_{i}(t_{i+1}), c_{s, i}(t_{i+1}))^{\frac{1}{p}}\mathrm{d}s\\
		\leq e^{-\frac{\lambda}{p}(t-t_i)} \int_I 
		V(c_{i}(t_{i}), c_{s, i}(t_{i}))^{\frac{1}{p}}\mathrm{d}s.\\
	\end{IEEEeqnarray*}
	By induction, it follows that, on any time period,
	\begin{IEEEeqnarray*}{rCl}
		\int_I V(c_{i+j}(t), c_{s, i+j}(t))^{\frac{1}{p}}\mathrm{d}s &<&  
		 e^{-\frac{\lambda}{p}(t-t_i)} \\&& \int_I 
		V(c_{i}(t_{i}), c_{s, i}(t_{i}))^{\frac{1}{p}}\mathrm{d}s.\\
	\end{IEEEeqnarray*}
	This gives, for all $t$,
	\begin{IEEEeqnarray*}{rCl}
		d(x^*(t), x(t)) &\leq& \int_I F(c_{i + j}(t), c_{s, i+j}(t))\mathrm{d}s\\
		&\leq& c_1^{-\frac{1}{p}} \int_I V(c_{i+j}(t), c_{s, i+j}(t))^\frac{1}{p}
		\mathrm{d}s\\
		&<&  c_1^{-\frac{1}{p}} e^{-\frac{\lambda}{p}(t-t_i)} \int_I 
		V(c_{i}(t_{i}), c_{s, i}(t_{i}))^{\frac{1}{p}}\mathrm{d}s\\
		&\leq&\left(\frac{c_2}{c_1}\right)^{-p} e^{-\frac{\lambda}{p}(t-t_i)} 
		\int_I F(c_i(t_i), c_{s, i}(t_i))\mathrm{d}s.\\
	\end{IEEEeqnarray*}
	This proves universal expontential stabilizability for $\alpha(V) =
	-\lambda V$.
\end{proof}

\begin{example}
	Consider the case $\mathcal{M} = \mathbb{R}^n$ with a Finsler-Lyapunov function 
	$V(x, \delta_x) = \delta_x\tran M
	\delta_x$ for some matrix $M$.  Suppose there exist bounds $c_1, c_2$ such
	that, for the Finsler structure $F = \sqrt{\delta_x\tran \delta_x}$, $c_1
	F(x, \delta_x)^2 \leq V(x, \delta_x) \leq c_2 F(x, \delta_x)^2$.  $F$ gives a
	Riemannian structure on $\mathbb{R}^n$.  If $V$ satisfies (\ref{th1}) for
	the system (\ref{dynamics}), (\ref{differential}), $M$ is a control
	contraction metric in the sense of \cite[Th. 1]{Manchester2017}.  In this
	setting, the construction of Section~\ref{CLlambda} is a generalization of 
	the sampled data controller given in \cite{Manchester2017} in two senses.
	Firstly, it provides universal stabilizability under the weaker dissipation
	condition $\alpha = 0$.  Secondly, it does not require computation of
	minimising geodesics - any initial path $c_i$ may be used, and the only
	condition on subsequent paths is that they are no longer than the forward
	image of the previous path (at the same time $t$).  This allows, for example,
	the path to be refined at sample points via a local search for a shorter
	path, without requiring the solution of a global shortest path.
\end{example}

\section{Conclusions}\label{conclusions}
\enlargethispage{-1.2in}
This work generalizes Control Contraction Metrics to Finsler manifolds.  This allows
a larger class of metrics to be used to measure distances between trajectories,
increasing the class of problems to which CCM methods can be applied.

The sampled data controllers constructed in Section~\ref{closed} do no require
computation of shortest paths between points (minimising geodesics).  This allows the
controllers to be applied in cases where minimising geodesics either do not exist or
are too computationally expensive to compute.

Further work remains to be done and will be the subject of future papers.  The open
loop controller constructed in Section~\ref{open} requires the Finsler-Lyapunov
function to be smooth.  This precludes the use of certain desirable metrics.
In
addition, the conditions under which the search for a non-Riemannian CCM is convex
are still to be determined.
\bibliographystyle{IEEEtran}
\bibliography{paper}{}

\end{document}